\documentclass{llncs}
\usepackage{times}

\usepackage{amssymb}
\usepackage{mathpartir}
\usepackage[dvips]{graphicx}
\usepackage{subfigure}
\usepackage{color}
\usepackage{amsmath}
\usepackage{url}

\definecolor{Violet}{rgb}{0.2,0.2,0.6}
\newcommand{\hwrp}{\mathtt{>}\!\!\mathtt{>}\!\mathtt{=}}

\newcommand{\type}[1]{\mathbf{#1}}
\newcommand{\unit}{()}

\newcommand{\dom}{\mathtt{dom}}
\newcommand{\seq}[1]{\overrightarrow{#1}}

\newenvironment{defn}{\begin{tabbing}
  \hspace{1.5em} \= \hspace{.10\linewidth - 1.5em} \= \hspace{1.5em} \= \kill
  }{
  \end{tabbing}}

\newcommand{\entry}[2]{\>$#1$\>\>#2}
\newcommand{\clause}[2]{$#1$\>\>#2}
\newcommand{\category}[2]{\clause{#1::=}{#2}}

\begin{document}

\title{Event Synchronization by Lightweight Message Passing}

\author{Avik Chaudhuri} 
\institute{Computer Science Department\\ University of California, Santa Cruz \\
\email{avik@cs.ucsc.edu}
}

\maketitle

\begin{abstract}
Concurrent ML's events and event combinators facilitate modular concurrent programming with first-class synchronization abstractions. A standard implementation of these abstractions relies on fairly complex manipulations of first-class continuations in the underlying language. In this paper, we present a lightweight implementation of these abstractions  in Concurrent Haskell, a language that already provides  first-order message passing. At the heart of our implementation is a new distributed synchronization protocol. In contrast with several previous translations of event abstractions in concurrent languages, we remain faithful to the standard semantics for events and event combinators; for example, we retain the symmetry of $\mathtt{choose}$ for expressing selective communication. 
\end{abstract}




\section{First-class synchronization abstractions}\label{intro}
In his doctoral thesis \cite{reppyThesis}, Reppy invents the concept of first-class synchrony to facilitate modular concurrent programming in ML. He argues: 
\begin{quote}
{\small Unfortunately there is a fundamental conflict between the desire for abstraction and the need for selective communication [in concurrent programs]. \ldots To resolve the conflict \ldots requires introducing a new abstraction mechanism that preserves the synchronous nature of the abstraction.}   \end{quote}
Thus, Reppy introduces a new type constructor $\type{event}$ to type synchronous operations in much the same way as $\rightarrow$ (``arrow") types functional values. 
\begin{quote}
{\small This allows us to represent synchronous operations as first-class values, instead of merely as functions \ldots [and design] a collection of combinators for defining
new $\type{event}$ values [from primitive ones]. \ldots Selective communication is expressed as a choice among $\type{event}$ values, which means
that user-defined abstractions can be used in a selective communication without breaking the abstraction.}
\end{quote}
Reppy implements events in an extension of ML, called Concurrent ML (CML)~\cite{reppyBook}, and provides a formal semantics for synchronization of events \cite{reppyThesis}. While the implementation itself is fairly complex, it allows programmers to write sophisticated communication and synchronization protocols as first-class abstractions in the resulting language. Next, we provide a brief introduction to CML's events and event combinators. The interested reader can find a more detailed account of first-class synchrony and its significance as a programming paradigm for concurrency in \cite{reppyBook}. 

In particular, note that $\type{channel}$ and $\type{event}$ are polymorphic type constructors in CML, as follows:
\begin{itemize}
\item The type $\type{channel}~\tau$ is given to channels that carry values of type $\tau$. 
\item The type $\type{event}~\tau$ is given to events that return values of type $\tau$ on synchronization. 
\end{itemize}
The combinators $\mathtt{receive}$ and $\mathtt{transmit}$ can build primitive events for synchronous communication. 
\vspace{-4mm}
\begin{eqnarray*}
\mathtt{receive} & : & \type{channel}~\tau \rightarrow \type{event}~\tau  \\
\mathtt{transmit} & : & \type{channel}~\tau \rightarrow \tau \rightarrow \type{event}~\type{\unit} 
\end{eqnarray*}
\begin{itemize}
\item $\mathtt{receive}~c$ returns an event that, on synchronization, accepts a message $M$ on channel $c$ and returns $M$. Such an event must synchronize with $\mathtt{transmit}\:c\:M$.
\item $\mathtt{transmit}\:c\:M$ returns an event that, on synchronization, sends the message $M$ on channel $c$ and returns $\unit$ (``unit"). Such an event must synchronize with  $\mathtt{receive}\:c$.
\end{itemize}
The combinator $\mathtt{choose}$ can non-deterministically select an event from a list of events, so that the selected event can be synchronized. In particular, $\mathtt{choose}$ can express any selective communication. 

The combinator $\mathtt{wrapabort}$ can specify an action that is spawned if an event is not selected by a $\mathtt{choose}$. \vspace{-1mm}
\begin{eqnarray*}
\mathtt{choose} & : & \type{[}\type{event}~\tau\type{]} \rightarrow \type{event}~\tau \\
\mathtt{wrapabort} & : & (\type{\unit } \rightarrow \type{\unit }) \rightarrow \type{event}~\tau \rightarrow \type{event}~\tau 
\end{eqnarray*}
\begin{itemize}
\item $\mathtt{choose}\:V$ returns an event that, on synchronization, synchronizes one of the events in list $V$ and ``aborts" the other events. 
\item $\mathtt{wrapabort}\:f\:v$ returns an event that, on synchronization, synchronizes the event $v$, and on abortion, spawns a thread that runs the code $f\:\unit$. Here, if $v$ itself is of the form $\mathtt{choose}\:V$ and one of the events in $V$ is selected, then $v$ is considered selected, so $f$ is \emph{not} spawned. 
\end{itemize}
The combinators $\mathtt{guard}$ and $\mathtt{wrap}$ can specify actions that are run before and after synchronization, respectively. \vspace{-2mm}
\begin{eqnarray*}
\mathtt{guard} & : & (\type{\unit } \rightarrow \type{event}~\tau) \rightarrow \type{event}~\tau \\
\mathtt{wrap} & : & \type{event}~\tau \rightarrow (\tau \rightarrow \tau') \rightarrow \type{event}~\tau' 
\end{eqnarray*}
\begin{itemize}
\item $\mathtt{guard}\:f$ returns an event that, on synchronization, synchronizes the event returned by the code $f\:\unit$. Here, $f\:\unit$ is run every time a thread tries to synchronize $\mathtt{guard}\:f$.
\item $\mathtt{wrap}\:v\:f$ returns an event that, on synchronization, synchronizes the event $v$ and applies function $f$ to the result. 
\end{itemize}
Finally, the function $\mathtt{sync}$ can synchronize an event and return the result.  \vspace{-1mm}
$$\mathtt{sync} ~~:~~  \type{event}~\tau \rightarrow \tau  \vspace{-1mm}$$
Note that by construction, an event can synchronize at exactly one ``commit point", where a message is either sent or accepted on a channel. This commit point may be selected among several other, potential commit points. Some code may be run before synchronization, as specified by $\mathtt{guard}$ functions throughout the event. Some more code may be run after synchronization, as specified by $\mathtt{wrap}$ functions that surround the commit point, and by $\mathtt{wrapabort}$ functions that do not surround the commit point. 

Reppy's implementation of synchronization involves fairly complex manipulations of first-class continuations in phases 
\cite{reppyBook}. 
Even the channel communication functions
\begin{eqnarray*}
\mathtt{accept} & : & \type{channel}~\tau \rightarrow \tau \\
\mathtt{send} & : & \type{channel}~\tau \rightarrow \tau \rightarrow \type{\unit}
\end{eqnarray*}
are derived by synchronization on the respective base events.
\begin{eqnarray*}
\mathtt{accept}\:c & = & \mathtt{sync}\:(\mathtt{receive}\:c) \\
\mathtt{send}\:c\:M & = & \mathtt{sync}\:(\mathtt{transmit}\:c\:M)
\end{eqnarray*}
In contrast, in this paper we show how to implement first-class event synchronization in a language that already provides first-order synchronous communication. 
Our implementation relies on a new distributed synchronization protocol, which  
we formalize as an abstract state machine and prove correct (Section~\ref{sec:protocol}). We concretely implement this protocol by message passing in Concurrent Haskell \cite{conchaskell}, a language that is quite close to the pi calculus \cite{milner89calculus}.
Building on this implementation, we present an encoding of CML events and event combinators in Concurrent Haskell (Sections \ref{compile} and \ref{sec:wrapabort}). 

We are certainly not the first to encode CML-style concurrency primitives in another language. However, a lightweight implementation of first-class event synchronization by message passing, in the exact sense of Reppy~\cite{reppyThesis}, has not appeared before. We defer a more detailed discussion on related work to Section~\ref{sec:relwork}.


Before we present our protocol, we introduce its main ideas through the following (fictional) narrative, which describes an analogous protocol for arranging marriages. 
%
\begin{quote}
In the Land of Fr\={\i}g, there are many young inhabitants who are eager to get married. At the same time, the rivalry among siblings is so fierce that if a young man or woman gets married, his or her siblings commit suicide. There are many priests, who serve as matchmakers. Eager young inhabitants flock to these matchmakers to meet prospective partners of the opposite sex. When two such partners meet, they reserve the priest who matches them, and then inform their parents. Meanwhile, the priest stops matching other couples. 

\hspace{5.5mm}Parents select the first child to inform them about meeting a partner, and send their approval to the concerned priest. For all other children who are too late, on the other hand, they send back their refusal. If a priest receives approval from both parties, he confirms the marriage date to both sides; following this information, the couple weds. If one party refuses and the other approves, the priest alerts the approving side that the impending marriage must be canceled; following this information, the young members of that family begin searching for partners once again. The priest now resumes matching other couples.
\end{quote}
Obviously, we would like to have progress in the Land of Fr\={\i}g---\emph{weddings should be possible as long as there remain eligible couples}. Can we prove this property?
Yes. We first prove the following lemma: \emph{if there remain two inhabitants of opposite sex who are looking for partners, a priest can eventually match them}. 
Indeed, pick any priest. Either that priest is free, or two partners who have already met have reserved that priest. In the latter case, both partners inform their parents; both parents send their decisions to the priest; therefore, the priest is eventually free.

Now, suppose that there remain two inhabitants $A$ and $B$ of opposite sex who are looking for partners. Then (by the lemma above) a priest eventually matches them. Next, $A$ and $B$ inform their parents. Now, if both $A$ and $B$ are the first among their siblings to inform their parents, they eventually get married and we are done. On the other hand, suppose that one of $A$'s siblings informs $A$'s parent first. (Clearly, that sibling could not have been married when $A$ was looking for partners, since otherwise, $A$ would have been dead, not looking for partners.) Now, either that sibling gets married, and we are done; or, that sibling does not get married, and $A$ tries again. Similarly, either one of $B$'s siblings gets married and we are done, or $B$ tries again. Now, if $A$ and $B$ both try again, they can be the first among their siblings to inform their parents; so, they eventually get married, and we are done.


\section{A distributed event synchronization protocol}\label{sec:protocol}
We now present a distributed protocol for synchronizing events, that is based on the protocol for arranging marriages in the Land of Fr\={\i}g. 
Specifically, we interpret
\begin{itemize}
\item young inhabitants as potential commit points, or simply, \emph{points};
\item priests as channels;
\item parents as synchronization sites, or simply, \emph{synchronizers}.
\end{itemize}
In other words, points, channels, and synchronizers are principals in our protocol. 
A point is a site of pending input or output on a channel---every $\mathtt{receive}$ or $\mathtt{transmit}$ event contains a point. Every application of $\mathtt{sync}$ contains a synchronizer. 

We focus on synchronization of events that are built with the combinators $\mathtt{receive}$, $\mathtt{transmit}$, and $\mathtt{choose}$. The other combinators do not fundamentally affect the protocol; we consider them only in the concrete implementation in Section \ref{compile}. 
\subsubsection{A source language}
For brevity, we simplify the syntax of the source language.
\begin{itemize}
\item Actions $\alpha,\beta$ are of the form $c$ or $\overline c$ (input or output on channel $c$).
\item Programs are of the form $S_1~|~\dots~|~S_n$ (parallel composition of $S_1,\dots,S_n$), where each $S_i$ is either an action $\alpha$, or a synchronization of a choice of actions $\mathsf{select}(\seq{\alpha})$.
\end{itemize}
Further, we consider only the following local reduction rule, which models selective communication: \vspace{-2mm}
$$\infer
	{c \in \seq \alpha \\ \overline c \in \seq{\beta}}
	{\mathsf{select}(\seq \alpha)~|~\mathsf{select}(\seq{\beta}) \rightarrow c~|~\overline c}
$$
besides the usual structural rules for parallel composition. In particular, we ignore reduction of actions at this level of abstraction.

\subsubsection{A distributed abstract state machine for synchronization}
We formalize our protocol as a distributed abstract state machine that implements the above semantics of selective communication. 
Let $\sigma$ range over states of the machine. These states may be of the form $\sigma~|~\sigma'$ (parallel composition), $(\nu p)~\sigma$ (name restriction), or $\varsigma$ (sub-state of some principal in the protocol). Sub-states may be of the following forms.

\parbox{14cm}{
\hspace{-5mm}
\parbox{7cm}{
\begin{defn}
\category{\varsigma_p}{sub-states of points}\\
\entry{p \mapsto \alpha}{unmatched}\\
\entry{\mathsf{Candidate}_p}{matched}\\
\entry{\alpha}{married}\\
\category{\varsigma_c}{sub-states of channels}\\
\entry{\odot_c}{free}\\
\entry{\mathsf{Match}_c(p,q)}{busy}
\end{defn}
}
\quad
\parbox{5cm}{
\begin{defn}
\category{\varsigma_s}{sub-states of synchronizers}\\
\entry{\Box_s}{open}\\
\entry{\boxtimes_s}{closed}\\
\entry{\mathsf{Select}_s(p)}{approved}\\
\entry{\mathsf{Reject}(p)}{refused}\\
\entry{\mathsf{Done}_s(p)}{confirmed}\\
\entry{\mathsf{Retry}_s}{canceled}
\end{defn}
}
}
%
Here $p$, $c$, and $s$ range over points, channels, and synchronizers. A synchronizer is a partial function from points to actions; we represent this function as a parallel composition of bindings of the form $p \mapsto \alpha$. Further, we require that each point belongs to exactly one synchronizer, that is, for any $s$ and $s'$, $s \neq s' ~\Rightarrow~ \dom(s) \cap \dom(s') = \varnothing$. The semantics of the machine is described by the following local transition rules, plus the usual structural rules for parallel composition and name restriction (\emph{cf.} the pi calculus \cite{milner89calculus}, for example). In Section \ref{compile}, these rules are implemented by message passing between appropriate processes run at points, channels, and synchronizers.
$$\infer
	{}
	{p \mapsto c~|~q \mapsto \overline c~|~\odot_c \rightarrow \mathsf{Candidate}_p~|~\mathsf{Candidate}_{q}~|~\mathsf{Match}_c(p,q)}
\qquad ({\rm I})	
$$
$$\infer
	{p \in \dom(s)}
	{\mathsf{Candidate}_p~|~\Box_s \rightarrow \boxtimes_s~|~\mathsf{Select}_s(p)}
~ ({\rm II.i}) \quad
\infer
	{p \in \dom(s)}
	{\mathsf{Candidate}_p~|~\boxtimes_s \rightarrow \boxtimes_s~|~\mathsf{Reject}(p)}
~ ({\rm II.ii})	
$$
$$\infer
	{}
	{\mathsf{Select}_s(p)~|~\mathsf{Select}_{s'}(q)~|~\mathsf{Match}_c(p,q) \rightarrow \mathsf{Done}_s(p)~|~\mathsf{Done}_{s'}(q)~|~\odot_c}
\qquad ({\rm III.i})	
$$
$$\infer
	{}
	{\mathsf{Select}_s(p)~|~\mathsf{Reject}(q)~|~\mathsf{Match}_c(p,q) \rightarrow \mathsf{Retry}_s~|~\odot_c}
\qquad ({\rm III.ii})	
$$
$$\infer
	{}
	{\mathsf{Reject}(p)~|~\mathsf{Select}_{s}(q)~|~\mathsf{Match}_c(p,q) \rightarrow \mathsf{Retry}_s~|~\odot_c}
\qquad ({\rm III.iii})	
$$
$$\infer
	{}
	{\mathsf{Reject}(p)~|~\mathsf{Reject}(q)~|~\mathsf{Match}_c(p,q) \rightarrow \odot_c}
\qquad ({\rm III.iv})	
$$
$$\infer
	{s(p) = \alpha}
	{\mathsf{Done}_s(p) \rightarrow \alpha}
\quad ({\rm IV.i}) \quad~~
\infer
	{\dom(s) = \seq p}
	{\mathsf{Retry}_s \rightarrow (\nu \seq p)~(\Box_s~|~s)}
\quad ({\rm IV.ii}) \vspace{1mm}
$$
The rules may be read as follows. 
\begin{description}
\item[{\rm (I)}] Two points $p$ and $q$, bound to complementary actions on channel $c$, react with $c$ if it is free ($\odot_c$), so that $p$ and $q$ become candidates and the channel becomes busy. 
\item[{\rm (II.i--ii)}] Next, $p$ (and likewise, $q$) reacts with its synchronizer $s$. If the synchronizer is open ($\Box_s$), it now becomes closed ($\boxtimes_s$), and $p$ is declared selected by $s$. If the synchronizer is already closed, then $p$ is rejected. 
\item[{\rm (III.i--iv)}] If both $p$ and $q$ are selected, $c$ confirms the selections to both parties. If only one of them is selected, $c$ cancels that selection. The channel now becomes free. 
\item[{\rm (IV.i--ii)}] If the selection of $p$ is confirmed, the action bound to $p$ is released. 
Otherwise, the synchronizer ``reboots" with fresh names for the points in its domain. 
\end{description}

\subsubsection{Compilation} 
We compile the source language to this machine. Let the symbol $\Pi$ denote finite parallel composition. Suppose that the set of channels in a program $\Pi_{i \in 1..n} S_i$ is $\mathcal C$. We compile this program to the state $\displaystyle\Pi_{c \in \mathcal C} \odot_c~|~\Pi_{i\in 1..n} \widehat{S_i}$,
where
\[
\widehat S = \left\{
\begin{array}{ll}
\alpha  & \mbox{ if }S = \alpha  \\
(\nu \seq{p_j})~(\Box_s~|~s) & \mbox{ if }S = \mathsf{select}(\seq {\alpha_j})\mbox{, where }s = \Pi_{j}~(p_j \mapsto \alpha_j)\mbox{ for fresh names }\seq{p_j}
\end{array}
\right.
\]

\subsubsection{Correctness}
We prove that our protocol is correct, that is, the abstract machine correctly implements selective communication, by showing that the compilation from programs to states preserves progress and safety. Let a \emph{denotation} be a list of actions. 
The denotations of programs and states are derived by the function $\ulcorner \cdot\urcorner$, as follows. \vspace{-3mm}

\parbox{12cm}{
\parbox{6cm}{
\begin{eqnarray*}
\ulcorner S_1~|~\dots~|~S_n \urcorner & = & \ulcorner S_1 \urcorner \uplus \dots \uplus \ulcorner S_n \urcorner \\
\ulcorner \alpha \urcorner & = & [\alpha] \\
\ulcorner \mathsf{select}(\seq \alpha) \urcorner & = & []
\end{eqnarray*}
}
\parbox{5cm}{
\begin{eqnarray*}
\ulcorner \sigma~|~\sigma' \urcorner & = & \ulcorner \sigma \urcorner \uplus \ulcorner \sigma' \urcorner \\
\ulcorner (\nu p)~\sigma \urcorner & = & \ulcorner \sigma \urcorner \\
\ulcorner \varsigma \urcorner & = & 
\left\{
\begin{array}{ll}
[\alpha] & \mbox{ if }\varsigma = \alpha\\
\mbox{$[]$} & \mbox{ otherwise }
\end{array}
\right.
\end{eqnarray*}
}
}
Now, if a program is compiled to some state, then the denotations of the program and the state coincide. Further, we have the following theorem, proved in the appendix.
\begin{theorem}[Correctness]\label{bisim} Let $\mathcal C$ be the set of channels in a program $\Pi_{i \in 1..n} S_i$. Then $\Pi_{i \in 1..n} S_i~\sim~\displaystyle\Pi_{c \in \mathcal C} \odot_c~|~\Pi_{i\in 1..n} \widehat{S_i}$, where $\sim$ is the largest relation such that 
$\mathcal P \sim \sigma$ iff
\begin{description}
\item[{\em (Correspondence)}] $\sigma \rightarrow^\star \sigma'$ for some $\sigma'$ such that $\ulcorner \mathcal P \urcorner = \ulcorner \sigma' \urcorner$;
\item[{\em (Safety)}] if $\sigma \rightarrow \sigma'$ for some $\sigma'$, then $\mathcal P \rightarrow^\star \mathcal P'$ for some $\mathcal P'$ such that $\mathcal P' \sim \sigma'$;
\item[{\em (Progress)}] if $\mathcal P \rightarrow \_$, then $\sigma \rightarrow^+\! \sigma'$ and $\mathcal P \rightarrow \mathcal P'$ for some $\sigma'$ and $\mathcal P'$ such that  $\mathcal P' \sim \sigma'$.\vspace{-1mm}
\end{description}
\end{theorem}

%

\subsubsection{Example}
Consider the program $\mathsf{select}(\overline x,\overline y)~|~\mathsf{select}(y,z)~|~\mathsf{select}(\overline z)~|~\mathsf{select}(x)$. The program can reduce either to $\overline x~|~z~|~\overline z~|~x$, or to $\overline y~|~y~|~\mathsf{select}(\overline z)~|~\mathsf{select}(x)$. The denotations of these reduced programs are $\{\overline x,x,\overline z,z\}$ and $\{\overline y,y\}$, respectively. The original program is compiled to the state
\[
\left.
\begin{array}{lcl}
\odot_x~|~\odot_y~|~\odot_z~|~&&(\nu p_{\bar x} p_{\bar y})~(\Box_{(p_{\bar x} \mapsto \overline x~|~p_{\bar y} \mapsto \overline y)}~|~p_{\bar x} \mapsto \overline x~|~p_{\bar y} \mapsto \overline y)~|~\\
&&(\nu p_{y} p_{z})~(\Box_{(p_{y} \mapsto y~|~p_{z} \mapsto z)}~|~p_y \mapsto y~|~p_z \mapsto z)~|~\\
&&(\nu p_{\bar z})~(\Box_{(p_{\bar z} \mapsto \overline z)}~|~p_{\bar z} \mapsto \overline z)~|~\\
&&(\nu p_{x})~(\Box_{(p_{x} \mapsto x)}~|~p_{x} \mapsto x)
\end{array}
\right.
\]
This state can transition in multiple steps to either of the following states, with denotations $\{\overline x,x,\overline z,z\}$ and $\{\overline y,y\}$, respectively. (In these states, $\sigma_{\it junk}$ can be garbage-collected, and is separated out for readability.)
\begin{itemize}
\item $
\overline x~|~z~|~\overline z~|~x~|~
\odot_x~|~\odot_y~|~\odot_z~|~\sigma_{\it junk}
$
\item $\overline y~|~y~|~(\nu p_{\bar z})~(\Box_{(p_{\bar z} \mapsto \overline z)}~|~p_{\bar z} \mapsto \overline z)~|~(\nu p_{x})~(\Box_{(p_{x} \mapsto x)}~|~p_{x} \mapsto x)~|~
\odot_x~|~\odot_y~|~\odot_z~|~\sigma_{\it junk}
$
\end{itemize}
 $$\sigma_{\it junk}~\triangleq~(\nu p_{\bar x} p_{\bar y} p_{y} p_{z} p_{\bar z} p_{x})~(\boxtimes_{(p_{\bar x} \mapsto \overline x~|~p_{\bar y} \mapsto \overline y)}~|~\boxtimes_{(p_{y} \mapsto y~|~p_{z} \mapsto z)}~|~\boxtimes_{(p_{\bar z} \mapsto \overline z)}~|~\boxtimes_{(p_{x} \mapsto x)})$$

\section{A concrete implementation in Concurrent Haskell}\label{compile}
The abstract machine of the previous section can be concretely implemented by a system of communicating processes. Indeed, we now present a complete implementation of a CML-style event library in a fragment of Concurrent Haskell with first-order message passing. This fragment is rather close to the pi calculus. Thus, we ensure that our implementation can be ported without difficulty to languages that support first-order communication. At the same time, we take advantage of Haskell's type system to show how events and event combinators can be typed under the $\type{IO}$ monad \cite{andythesis,impfunprog}.

Before we proceed, we briefly review some of the concurrency primitives in Concurrent Haskell. 
Note that $\type{MVar}$ and $\type{IO}$ are polymorphic type constructors, as follows:
\begin{itemize}
\item The type $\type{MVar}~\tau$ is given to a communication cell that carries values of type $\tau$.
\item The type $\type{IO}~\tau$ is given to a computation that yields results of type $\tau$, with possible side effects via communication.
\end{itemize}
We rely on the following semantics of $\type{MVar}$ cells.
\begin{itemize}
\item A cell can carry at most one value at a time.
\item The function $\mathtt{New} :: \type{IO}~(\type{MVar}~\tau)$ returns a fresh, empty cell.
\item The function $\mathtt{Get}:: \type{MVar}~\tau \rightarrow \type{IO}~\tau$ is used to read from a cell; $\mathtt{Get}\:m$ blocks if the cell $m$ is empty, else gets the content of $m$ (thereby emptying it).
\item The function $\mathtt{Put}:: \type{MVar}~\tau \rightarrow \tau \rightarrow \type{IO}~\type{\unit }$ is used to write to a cell; $\mathtt{Put}\:m\:M$ puts the term $M$ in cell $m$ if it is empty, else blocks.
\end{itemize}
Further, we rely on the following semantics of $\type{IO}$ computations.
\begin{itemize}
\item The function $\mathtt{fork} :: \type{IO}~\type{\unit } \rightarrow \type{IO}~\type{\unit }$ is used to spawn a concurrent computation.
\item The function $\mathtt{return} :: \tau \rightarrow \type{IO}~\type{\tau}$ is used to inject a value into a computation.
\item Computations can be sequentially composed by ``piping". We use Haskell's convenient $\mathtt{do}\{\dots;\dots\}$ notation for this purpose, instead of applying the \emph{de jure}  function $\_\:\hwrp\:\_ :: \type{IO}~\tau \rightarrow (\tau \rightarrow \type{IO}~\tau') \rightarrow \type{IO}~\tau'$.
\end{itemize}
%
%

\subsubsection{Implementing synchronization by message passing}
We implement the following functions for programming with first-class events in Concurrent Haskell. (Note the differences between ML and Haskell types for these functions. Since Haskell is purely functional, we must embed types for computations with possible side-effects via communication, within the $\type{IO}$ monad. Further, since evaluation in Haskell is lazy, we can discard abstractions that only serve to ``delay" evaluation.)
\begin{eqnarray*}
\mathtt{new} & :: & \type{IO}~(\type{channel}~\tau) \\
\mathtt{receive} & :: & \type{channel}~\tau \rightarrow \type{event}~\tau \\
\mathtt{transmit} & :: & \type{channel}~\tau \rightarrow \tau \rightarrow \type{event}~\type{\unit } \\
\mathtt{guard} & :: & \type{IO}~(\type{event}~\tau) \rightarrow \type{event}~\tau \\
\mathtt{wrap} & :: & \type{event}~\tau \rightarrow (\tau \rightarrow \type{IO}~\tau') \rightarrow \type{event}~\tau' \\
\mathtt{choose} & :: & \type{[}\type{event}~\tau\type{]} \rightarrow \type{event}~\tau \\
\mathtt{wrapabort} & :: & \type{IO}~\type{\unit } \rightarrow \type{event}~\tau \rightarrow \type{event}~\tau \\
\mathtt{sync} & :: & \type{event}~\tau \rightarrow \type{IO}~\tau
\end{eqnarray*}
%
For now, 
we focus on events that are built without $\mathtt{wrapabort}$ (\emph{i.e.}, we focus on programs without abort actions); the full implementation appears in Section \ref{sec:wrapabort}. 

We begin by concretizing the abstract state machine in Section \ref{sec:protocol}. Specifically, we run some protocol code at points, channels, and synchronizers, which reduce by simple message passing on $\type{MVar}$ cells. In this implementation:
%
\begin{itemize}
\item Each point is identified with a fresh name $p :: \type{Point}$.
\item Each channel $c$ is identified with a pair of fresh cells ${\it in}^{[c]} :: \type{In}$ and ${\it out}^{[c]} :: \type{Out}$ on which it receives messages from points that are bound to actions on $c$. 
\item Each synchronizer is identified with a fresh cell $s :: \type{Synchronizer}$ on which it receives messages from points in its domain.
\end{itemize}
%
Before we present protocol code, let us describe the sequence of messages exchanged in a typical session of the protocol, and mention the involved sub-states. On the way, we develop type definitions for the $\type{MVar}$ cells on which those messages are exchanged.
\begin{itemize}
\item A point $p$ (at state $p \mapsto c$ or $p \mapsto \overline c$) begins by sending a message to $c$ on its respective input or output cell ${\it in}^{[c]}$ or ${\it out}^{[c]}$; the message contains a fresh cell ${\it candidate}^{[p]} :: \type{Candidate}$ on which $p$ expects a reply from $c$.
$$\left.
\begin{array}{rcccl}
\mathtt{type}~~\type{In} & \!~=~\! & \type{MVar}~\type{Candidate} \\
\mathtt{type}~~\type{Out} & \!~=~\! & \type{MVar}~\type{Candidate} 
\end{array}
\right.\vspace{-1mm}
$$
%
\item When $c$ (at state $\odot_c$) gets a pair of messages on ${\it in}^{[c]}$ and ${\it out}^{[c]}$, say from $p$ and another point $q$, it replies by sending fresh cells ${\it decision}^{[p]} :: \type{Decision}$ and ${\it decision}^{[q]} :: \type{Decision}$ on ${\it candidate}^{[p]}$ and ${\it candidate}^{[q]}$ respectively (reaching state $\mathsf{Match}_c(p,q)$), and expects the synchronizers for $p$ and $q$ to reply on them. 
$$\mathtt{type}~~\type{Candidate}  ~=~  \type{MVar}~\type{Decision}$$
%
\item On receiving a message from $c$ on ${\it candidate}^{[p]}$, $p$ (reaching state $\mathsf{Candidate}_p$) tags the message with its name and forwards it to its synchronizer on the cell $s$.\vspace{-1mm}
$$\mathtt{type}~~\type{Synchronizer} ~=~ \type{MVar}~(\type{Point},\type{Decision})$$
%
\item If $p$ is the first point to send such a message on $s$ (that is, $s$ is at state $\Box_s$), a fresh cell ${\it commit}^{[p]} :: \type{Commit}$ is sent back on ${\it decision}^{[p]}$ (reaching state $\boxtimes_s~|~\mathsf{Select}_s(p)$); for each subsequent message received on $s$, say from $p'$, a blank message is sent back on ${\it decision}^{[p']}$ (reaching state $\boxtimes_s~|~\mathsf{Reject}(p')$). 
$$\mathtt{type}~~\type{Decision} ~=~\type{MVar}~(\type{Maybe}~\type{Commit})\footnote{Recall that the Haskell type $\type{Maybe}~\tau$ is given to a value that is either $\mathtt{Nothing}$, or of the form $\mathtt{Just}~v$ where $v$ is of type $\tau$. The function $\mathtt{maybe} :: \tau' \rightarrow (\tau \rightarrow \tau') \rightarrow \type{Maybe}~\tau \rightarrow \tau'$ is the associated case analyzer. For instance, the function $\mathtt{isJust} :: \type{Maybe}~\tau \rightarrow \type{Bool}$ is defined as $\mathtt{maybe}~\mathtt{False}~(\lambda \_.~\mathtt{True})$.}$$
%
\item On receiving messages from the respective synchronizers of $p$ and $q$ on ${\it decision}^{[p]}$ and ${\it decision}^{[q]}$, $c$ inspects the messages and responds (reaching state $\odot_c$).\vspace{2mm}
\begin{itemize}
\item If both ${\it commit}^{[p]}$ and ${\it commit}^{[q]}$ have come in, a positive signal is sent back on ${\it commit}^{[P]}$ and ${\it commit}^{[Q]}$.
\item If only ${\it commit}^{[p]}$ has come in, a negative signal is sent back on ${\it commit}^{[p]}$; if only ${\it commit}^{[q]}$ has come in, a negative signal is sent back on ${\it commit}^{[q]}$.
\end{itemize}
%
\vspace{2mm}
$$\mathtt{type}~~\type{Commit} ~=~ \type{MVar}~\type{Bool}$$
%
\item If $s$ receives a positive signal on ${\it commit}^{[p]}$ (reaching state $\mathsf{Done}_s(p)$), it signals on $p$ to continue. If, instead, a negative signal is received (reaching state  $\mathsf{Retry}_s$), another session ensues. \vspace{-3mm}
$$\mathtt{type}~~\type{Point} ~=~ \type{MVar}~\type{\unit }$$
\end{itemize}
\paragraph{Protocol code for points}
The protocol code run by points abstracts on a cell $s$ for the associated synchronizer, and a name $p$ for the point itself. Depending on whether the point is for input or output, the code additionally abstracts on an input cell $i$ or output cell $o$, and an input or output action $\alpha$.
\begin{eqnarray*}
&&\mathtt{@PointI} :: \type{Synchronizer} \rightarrow \type{Point} \rightarrow \type{In} \rightarrow \type{IO}~\tau \rightarrow \type{IO}~\tau\\ 
&&\mathtt{@PointI}\:s\:p\:i\:\alpha = \mathtt{do}~\{{\it candidate} \leftarrow \mathtt{New};~\mathtt{Put}\:i\:{\it candidate};\\
&&~~~~~~~~~~~~~~~~~~~~~~~~~~~~~~~~~~~~~\:\:\:{\it decision} \leftarrow \mathtt{Get}\:{\it candidate};~\mathtt{Put}\:s\:(p,{\it decision});\\
&&~~~~~~~~~~~~~~~~~~~~~~~~~~~~~~~~~~~~~\:\:\:\mathtt{Get}\:t;~\alpha\}\\
&&\mathtt{@PointO} :: \type{Synchronizer} \rightarrow \type{Point} \rightarrow \type{Out} \rightarrow \type{IO}~\unit \rightarrow \type{IO}~\unit\\ 
&&\mathtt{@PointO}\:s\:p\:o\:\alpha = \mathtt{do}~\{{\it candidate} \leftarrow \mathtt{New};~\mathtt{Put}\:o\:{\it candidate};\\
&&~~~~~~~~~~~~~~~~~~~~~~~~~~~~~~~~~~~~~\:\:\:{\it decision} \leftarrow \mathtt{Get}\:{\it candidate};~\mathtt{Put}\:s\:(p,{\it decision});\\
&&~~~~~~~~~~~~~~~~~~~~~~~~~~~~~~~~~~~~~\:\:\:\mathtt{Get}\:t;~\alpha\}
\end{eqnarray*}
%
\paragraph{Protocol code for channels}
The protocol code run by channels abstracts on an input cell $i$ and an output cell $o$ for the channel.
\begin{eqnarray*}
\!\!\!&\!&\mathtt{@Chan} :: \type{In} \rightarrow \type{Out} \rightarrow \type{IO}~\type{\unit }\\
\!\!\!&\!&\mathtt{@Chan}\:i\:o = \mathtt{do}~\{{\it candidate}_i \leftarrow \mathtt{Get}\:i;~{\it candidate}_o \leftarrow \mathtt{Get}\:o;\\
\!\!\!&\!&~~~~~~~~~~~~~~~~~~~~~~~~~~\:\:\:\:{\it decision}_i \leftarrow \mathtt{New};~\mathtt{Put}\:{\it candidate}_i\:{\it decision}_i;~x_i \leftarrow \mathtt{Get}\:{\it decision}_i;\\
\!\!\!&\!&~~~~~~~~~~~~~~~~~~~~~~~~~~\:\:\:\:{\it decision}_o \leftarrow \mathtt{New};~\mathtt{Put}\:{\it candidate}_o\:{\it decision}_o;~x_o \leftarrow \mathtt{Get}\:{\it decision}_o;\\
\!\!\!&\!&~~~~~~~~~~~~~~~~~~~~~~~~~~\:\:\:\:\mathtt{maybe}\:(\mathtt{return}\:\unit)\:(\lambda {\it commit}_i.~\mathtt{Put}\:{\it commit}_i\:(\mathtt{isJust}\:x_o))\:x_i; \\
\!\!\!&\!&~~~~~~~~~~~~~~~~~~~~~~~~~~\:\:\:\:\mathtt{maybe}\:(\mathtt{return}\:\unit)\:(\lambda {\it commit}_o.~\mathtt{Put}\:{\it commit}_o\:(\mathtt{isJust}\:x_i))\:x_o \}
\end{eqnarray*}
%
\paragraph{Protocol code for synchronizers}
The protocol code run by synchronizers abstracts on a cell $s$ for that synchronizer and some ``rebooting code" $X$. (Here, we encode a loop with the function $\mathtt{fix} :: (\tau \rightarrow \tau) \rightarrow \tau$; recall that $\mathtt{fix}~f$ reduces to $f~(\mathtt{fix}~f)$.) 
\begin{eqnarray*}
\!\!\!&\!&\mathtt{@Sync} :: \type{Synchronizer} \rightarrow \type{IO}~\type{\unit } \rightarrow \type{IO}~\type{\unit }\\
\!\!\!&\!&\mathtt{@Sync}\:s\:X = \mathtt{do}~\{(p,{\it decision}) \leftarrow \mathtt{Get}\:s;\\
\!\!\!&\!&~~~~~~~~~~~~~~~~~~~~~~~~~~~~\:\:\:\:\mathtt{fork}\:(\mathtt{fix}\:(\lambda {\it iter}.~\mathtt{do}~\{\\
\!\!\!&\!&~~~~~~~~~~~~~~~~~~~~~~~~~~~~\:\:\:\:~~~~~(p',{\it decision}') \leftarrow \mathtt{Get}\:s;~\mathtt{Put}\:{\it decision}'\:\mathtt{Nothing};~{\it iter}\}));\\
\!\!\!&\!&~~~~~~~~~~~~~~~~~~~~~~~~~~~~\:\:\:\:{\it commit} \leftarrow \mathtt{New};~\mathtt{Put}\:{\it decision}\:(\mathtt{Just}\:{\it commit}); \\
\!\!\!&\!&~~~~~~~~~~~~~~~~~~~~~~~~~~~~\:\:\:\:{\it done} \leftarrow \mathtt{Get}\:{\it commit};~\mathtt{if}\:{\it done}\:\mathtt{then}\:(\mathtt{Put}\:p\:\unit)\:\mathtt{else}\:X\}
\end{eqnarray*}
%
%
We instantiate these processes in the translation of $\mathtt{new}$, $\mathtt{receive}$, $\mathtt{transmit}$, and $\mathtt{sync}$ below.
But first, let us translate types for channels and events. 
\subsubsection{Translation of types}
The Haskell types for ML $\type{channel}$ and $\type{event}$ values are:
\begin{eqnarray*}
\type{channel}~\tau & = & (\type{In},\type{Out},\type{MVar}~\tau) \\
\type{event}~\tau & = & \type{Synchronizer} \rightarrow \type{IO}~\tau
\end{eqnarray*}
An ML $\type{channel}$ is a Haskell $\type{MVar}$ tagged with a pair of input and output cells. An ML $\type{event}$ is a Haskell $\type{IO}$ function that abstracts on a synchronizer cell. 

\subsubsection{Translation of functions}
We now translate functions for programming with events. 
We begin by compiling the ML function for creating channels.
\begin{eqnarray*}
&&\mathtt{new} :: \type{IO}~(\type{channel}~\tau) \\
&&\mathtt{new}  = \mathtt{do}~\{i \leftarrow \mathtt{New};~o \leftarrow \mathtt{New};\\
&&~~~~~~~~~~~~~~~~~~\:\mathtt{fork}\:(\mathtt{fix}\:(\lambda {\it iter}.~\mathtt{do}~\{ \mathtt{@Chan}\:i\:o; ~{\it iter} \}));\\
&&~~~~~~~~~~~~~~~~~~\:m \leftarrow \mathtt{New};~\mathtt{return}\:(i,o,m) \}
\end{eqnarray*}
\begin{itemize}
\item The term $\mathtt{new}$ spawns a looping instance of $\mathtt{@Chan}$ with a fresh pair of input and output cells, and returns that pair along with a fresh $\type{MVar}$ cell that carries messages for the channel.
\end{itemize}
Next, we compile the ML combinators for building base communication events. Recall that a Haskell event is an $\type{IO}$ function that abstracts on the cell of its synchronizer.
\begin{eqnarray*}
&&\mathtt{receive} ::  \type{channel}~\tau \rightarrow \type{event}~\tau \\
&&\mathtt{receive}\:(i,o,m) =  \lambda s.~\mathtt{do}~\{p \leftarrow \mathtt{New};~\mathtt{@PointI}\:s\:p\:i\:(\mathtt{Get}\:m)\} 
\end{eqnarray*}
\vspace{-5mm}
\begin{eqnarray*}
&&\mathtt{transmit} ::  \type{channel}~\tau \rightarrow \tau \rightarrow \type{event}~\type{\unit } \\
&&\mathtt{transmit}\:(i,o,m)\:M = \lambda s.~\mathtt{do}~\{p \leftarrow \mathtt{New};~\mathtt{@PointO}\:s\:p\:o\:(\mathtt{Put}\:m\:M)\} 
\end{eqnarray*}
\begin{itemize}
\item The term $\:\mathtt{receive}~c~s\:$ runs an instance of $\mathtt{@PointI}$ with the synchronizer $s$, a fresh name for the point, the input cell for channel $c$, and an action that inputs on $c$.
\item The term $\:\mathtt{transmit}\:c\:M\:s\:$ is symmetric; it runs an instance of $\mathtt{@PointO}$ with the synchronizer $s$, a fresh name for the point, the output cell for channel $c$, and an action that outputs term $M$ on~$c$.
\end{itemize}
Next, we compile the ML event combinators for specifying actions that are run before and after synchronization. 
\begin{eqnarray*}
&&\mathtt{guard} :: \type{IO}~(\type{event}~\tau) \rightarrow \type{event}~\tau \\
&&\mathtt{guard}\:f = \lambda s.~\mathtt{do}~\{v \leftarrow f;~v\:s\}
\end{eqnarray*}
\vspace{-5mm}
\begin{eqnarray*}
&&\mathtt{wrap} :: \type{event}~\tau \rightarrow (\tau \rightarrow \type{IO}~\tau') \rightarrow \type{event}~\tau' \\
&&\mathtt{wrap}\:v\:f = \lambda s.~\mathtt{do}~\{x \leftarrow v\:s;~f~x\} 
\end{eqnarray*}
\begin{itemize}
\item The term $\:\mathtt{guard}\:f\:s\:$ runs the computation $f$ and passes the synchronizer cell $s$ to the event returned by the computation.
\item The term $\:\mathtt{wrap}\:v\:f\:s\:$ passes the synchronizer cell $s$ to the event $v$ and pipes the returned value to function $f$. 
\end{itemize}
Next, we compile the ML combinator for choosing among a list of events. (Here, we encode recursion over a list with the function $\mathtt{fold} :: (\tau' \rightarrow \tau \rightarrow \tau') \rightarrow \tau' \rightarrow [\tau] \rightarrow \tau'$; recall that $\mathtt{fold}~f~x~[]$ reduces to $x$ and $\mathtt{fold}~f~x~[v,V]$ reduces to $\mathtt{fold}~f~(f~x~v)~V$.)
\begin{eqnarray*}
&&\mathtt{choose} :: \type{[}\type{event}~\tau\type{]} \rightarrow \type{event}~\tau \\
&&\mathtt{choose}\:V = \lambda s.~\mathtt{do}~\{{\it temp} \leftarrow \mathtt{New};\\
&&~~~~~~~~~~~~~~~~~~~~~~~~~~~~~~~~~\:\:\:\mathtt{fold}\:(\lambda \_~v.~\mathtt{fork}\:(\mathtt{do}~\{x \leftarrow v\:s;~\mathtt{Put}\:{\it temp}\:x\}))\:\unit\:V; \\
&&~~~~~~~~~~~~~~~~~~~~~~~~~~~~~~~~~\:\:\:\mathtt{Get}\:{\it temp} \}
\end{eqnarray*}
\begin{itemize}
\item The term $\:\mathtt{choose}\:V\:s\:$ spawns a thread for each event $v$ in $V$, passing the synchronizer $s$ to $v$; any value returned by one of these threads is collected in a fresh cell ${\it temp}$ and returned. 
\end{itemize}
Finally, we compile the ML function for event synchronization.
\begin{eqnarray*}
&&\mathtt{sync} :: \type{event}~\tau \rightarrow \type{IO}~\tau \\
&&\mathtt{sync}\:v = \mathtt{do}~\{{\it temp} \leftarrow \mathtt{New};\\
&&~~~~~~~~~~~~~~~~~~~~~\:\:\:\mathtt{fork}\:(\mathtt{fix}\:(\lambda {\it iter}.~\mathtt{do}~\{\\
&&~~~~~~~~~~~~~~~~~~~~~\:\:\:~~~~~s \leftarrow \mathtt{New};~\mathtt{fork}\:(\mathtt{@Sync}\:s\:{\it iter});~x \leftarrow v\:s;~\mathtt{Put}\:{\it temp}\:x \}));\\
&&~~~~~~~~~~~~~~~~~~~~~\:\:\:\mathtt{Get}\:{\it temp} \}
\end{eqnarray*}
\begin{itemize}
\item The term $\:\mathtt{sync}\:v\:$ recursively spawns an instance of $\mathtt{@Sync}$ with a fresh synchronizer $s$ and passes $s$ to the event $v$; any value returned by one of these instances is collected in a fresh cell ${\it temp}$ and returned. 
\end{itemize}

\section{Compiling abort actions}\label{sec:wrapabort}

The implementation of the previous section does not account for $\mathtt{wrapabort}$. We now show how $\mathtt{wrapabort}$ can be handled by slightly extending our notion of $\type{event}$.

Recall that 
abort actions (such as those specified by $\mathtt{wrapabort}$) are spawned only at events that do not enclose the commit point. Therefore, in an implementation of $\mathtt{wrapabort}$, it makes sense to name events with the sets of points they enclose. 
However, computing the set of points that an event encloses should not interfere with the dynamic semantics. In particular, for an event built with $\mathtt{guard}$, we cannot run the $\mathtt{guard}$ functions to compute the set of points that the event encloses. Thus, we refrain from naming events at compile time. Instead, we introduce events as principals in our protocol; each event is named \emph{in situ} by computing the list of points it encloses at runtime. This list is carried on a fresh cell ${\it name} :: \type{Name}$ for the event. 
$$\mathtt{type}~~\type{Name} ~=~ \type{MVar}~[\type{Point}]$$
Further, each synchronizer carries a fresh cell ${\it abort} :: \type{Abort}$ on which it receives $\mathtt{wrapabort}$ functions from events, tagged with the list of points they enclose.
$$\mathtt{type}~~\type{Abort} ~=~ \type{MVar}~([\type{Point}], \type{IO}~\type{\unit})$$
Protocol code run by points and channels remain the same. We add a handler for $\mathtt{wrapabort}$ functions to the protocol code run by synchronizers. Accordingly, the code now abstracts on an ${\it abort}$ cell. 
\begin{eqnarray*}
\mathtt{@Sync} & :: & \type{Synchronizer} \rightarrow \type{Abort} \rightarrow \type{IO}~\type{\unit } \rightarrow \type{IO}~\type{\unit }\\
\mathtt{@Sync}\:s\:{\it abort}\:X & = & \mathtt{do}~\{\dots;\\
&&~~~~~~~\mathtt{if}\:{\it done}\:\mathtt{then}\:\mathtt{do}~\{\dots;\\
&&~~~~~~~~~~~~~~~~~~~~~~~~~~~~~~~~~~~~~~~\mathtt{fix}\:(\lambda {\it iter}.~\mathtt{do}~\{\\
&&~~~~~~~~~~~~~~~~~~~~~~~~~~~~~~~~~~~~~~~~~~~~(P,f) \leftarrow \mathtt{Get}\:{\it abort};~\mathtt{fork}\:{\it iter};\\
&&~~~~~~~~~~~~~~~~~~~~~~~~~~~~~~~~~~~~~~~~~~~~\mathtt{if}\:p \in P\:\mathtt{then}\:\mathtt{return}\:\unit\:\mathtt{else}\:f\})\}\\
&&~~~~~~~\mathtt{else}\:\dots \}
\end{eqnarray*}
Here, after signaling the commit point $p$ to continue (as earlier), the synchronizer continues to accept abort code $f$ on ${\it abort}$; such code is spawned only if the list of points $P$, enclosed by the event that sends that code, does not include $p$.

The extended Haskell type for $\type{event}$ values is as follows. 
\begin{eqnarray*}
\type{event}~\tau & = & \type{Synchronizer} \rightarrow \type{Name}  \rightarrow \type{Abort} \rightarrow \type{IO}~\tau
\end{eqnarray*}
Now, an ML $\type{event}$ is a Haskell $\type{IO}$ function that abstracts on a synchronizer, an abort cell, and a name cell that carries the list of points the event encloses. 

The Haskell function $\mathtt{new}$ does not change. We highlight minor changes in the remaining translations. We begin with the functions $\mathtt{transmit}$ and $\mathtt{receive}$. An event built with either function is named by a singleton containing the name of that point. \vspace{-1mm}
\begin{eqnarray*}
\mathtt{transmit}\:(i,o,m)\:M & = & \lambda s\:{\it name}\:{\it abort}.~\mathtt{do}~\{\dots;~ \mathtt{fork}\:(\mathtt{Put}\:{\it name}\:[p]);~ \dots\} \\
\mathtt{receive}\:(i,o,m) & = & \lambda s\:{\it name}\:{\it abort}.~\mathtt{do}~\{\dots;~ \mathtt{fork}\:(\mathtt{Put}\:{\it name}\:[p]);~ \dots\} 
\end{eqnarray*}
The function $\mathtt{choose}$ becomes slightly lengthy. A fresh ${\it name}'$ cell is passed to each event in the list; the names of those events are concatenated to name the $\mathtt{choose}$ event.\vspace{-2mm}
\begin{eqnarray*}
\!\!\!&&\mathtt{choose}\:V = \lambda s\:{\it name}\:{\it abort}.~\mathtt{do}~\{\dots; \\ 
\!\!\!&&~~~~~~~~~~~~~~~~~~~~~~~~~~~~~~~~~~~~~~~~~~~~~~~~~~~~~~~~P \leftarrow \mathtt{fold}\:(\lambda P\:v.~\mathtt{do}~\{\\
\!\!\!&&~~~~~~~~~~~~~~~~~~~~~~~~~~~~~~~~~~~~~~~~~~~~~~~~~~~~~~~~~~~~~~~~~~~~~~~~{\it name}' \leftarrow \mathtt{New};\\
\!\!\!&&~~~~~~~~~~~~~~~~~~~~~~~~~~~~~~~~~~~~~~~~~~~~~~~~~~~~~~~~~~~~~~~~~~~~~~~~\mathtt{fork}\:(\mathtt{do}~\{x \leftarrow v\:s\:{\it name}'\:{\it abort};~\dots\});\\
\!\!\!&&~~~~~~~~~~~~~~~~~~~~~~~~~~~~~~~~~~~~~~~~~~~~~~~~~~~~~~~~~~~~~~~~~~~~~~~~P' \leftarrow \mathtt{Get}\:{\it name}';~\mathtt{Put}\:{\it name}'\:P';\\
\!\!\!&&~~~~~~~~~~~~~~~~~~~~~~~~~~~~~~~~~~~~~~~~~~~~~~~~~~~~~~~~~~~~~~~~~~~~~~~~\mathtt{return}\:(P' \uplus P)\})~[]\:V; \\
\!\!\!&&~~~~~~~~~~~~~~~~~~~~~~~~~~~~~~~~~~~~~~~~~~~~~~~~~~~~~~~~\mathtt{fork}\:(\mathtt{Put}\:{\it name}\:P);\\
\!\!\!&&~~~~~~~~~~~~~~~~~~~~~~~~~~~~~~~~~~~~~~~~~~~~~~~~~~~~~~~~\dots \} 
\end{eqnarray*}
We now compile the ML event combinator for specifying abort actions. \vspace{-2mm}
\begin{eqnarray*}
&&\mathtt{wrapabort} :: \type{IO}~\type{\unit} \rightarrow \type{event}~\tau \rightarrow \type{event}~\tau \\
&&\mathtt{wrapabort}\:f\:v = \lambda s\:{\it name}\:{\it abort}.~\mathtt{do}~\{\\
&&~~~~~~~~~~\mathtt{fork}\:(\mathtt{do}~\{P \leftarrow \mathtt{Get}\:{\it name};~\mathtt{Put}\:{\it name}\:P;~\mathtt{Put}\:{\it abort}\:(P,f)\});\\
&&~~~~~~~~~~v\:s\:{\it name}\:{\it abort} \}
\end{eqnarray*}
\begin{itemize}
\item The term $\:\mathtt{wrapabort}\:f\:v\:s\:{\it name}\:{\it abort}\:$ spawns a thread that reads the list of enclosed events $P$ on the cell ${\it name}$ and sends the function $f$ along with $P$ on the cell  ${\it abort}$; the synchronizer $s$ is passed to the event $v$ along with ${\it name}$ and ${\it abort}$.
\end{itemize}
The functions $\mathtt{guard}$ and $\mathtt{wrap}$ remain similar.\vspace{-2mm}
\begin{eqnarray*}
\mathtt{guard}\:f & = & \lambda s\:{\it name}\:{\it abort}.~\mathtt{do}~\{v \leftarrow f;~v\:s\:{\it name}\:{\it abort}\}\\
\mathtt{wrap}\:v\:f & = & \lambda s\:{\it name}\:{\it abort}.~\mathtt{do}~\{x \leftarrow v\:s\:{\it name}\:{\it abort}; f~x\} 
\end{eqnarray*}
Finally, in the function $\mathtt{sync}$, a fresh ${\it abort}$ cell is now passed to $\mathtt{@Sync}$, and a fresh ${\it name}$ cell is created for the event to be synchronized.\pagebreak
\begin{eqnarray*}
&&\mathtt{sync}\:v = \mathtt{do}~\{\dots; \\
&&~~~~~~~~~~~~~~~~~~~~~~~~\mathtt{fork}\:(\mathtt{fix}\:(\lambda {\it iter}.~\mathtt{do}~\{\\
&&~~~~~~~~~~~~~~~~~~~~~~~~~~~~~\dots; ~{\it name} \leftarrow \mathtt{New}; ~{\it abort} \leftarrow \mathtt{New};\\
&&~~~~~~~~~~~~~~~~~~~~~~~~~~~~~ \mathtt{fork}\:(\mathtt{@Sync}\:s\:{\it abort}\:{\it iter});~ x \leftarrow v\:s\:{\it name}\:{\it abort}; ~\dots\}));\\
&&~~~~~~~~~~~~~~~~~~~~~~~~\dots \}
\end{eqnarray*}

\section{Implementing communication guards}\label{guards}
Beyond the standard primitives for communication in CML, some previous implementations of events further consider \emph{guarded communication}.  We discuss how our implementation can be easily extended to handle such communication. Specifically, we require the following $\mathtt{receive}$ combinator, that can carry a communication guard.
$$\mathtt{receive} ::  \type{channel}~\tau \rightarrow (\tau \rightarrow \type{Bool}) \rightarrow \type{event}~\tau$$
Intuitively, $\:\mathtt{receive}~c~{\it cond}\:$ synchronizes with $\:\mathtt{transmit}~c~M\:$ only if $\:{\it cond}\:M\:$ is true. In our implementation, we make some slight adjustments to the types of some $\type{MVar}$ cells.
$$\left.
\begin{array}{rclcl}
\mathtt{type}~~\type{In}~\tau & \!~=~\! & \type{MVar}~(\type{Candidate},\tau \rightarrow \type{Bool}) \\
\mathtt{type}~~\type{Out}~\tau & \!~=~\! & \type{MVar}~(\type{Candidate}, \tau) \\ 
\mathtt{type}~~\type{Candidate} & \!~=~\! & \type{MVar}~(\type{Maybe}~\type{Decision}) 
\end{array}
\right.
$$
Next, we adjust the protocol code run by points and channels. Input and output points that are bound to actions on $c$ respectively send their conditions and messages to $c$. A pair of points is matched only if the message of one satisfies the condition of the other.
\begin{eqnarray*}
\!\!\!&\!&\mathtt{@Chan} :: \type{In}~\tau \rightarrow \type{Out}~\tau \rightarrow \type{IO}~\type{\unit }\\
\!\!\!&\!&\mathtt{@Chan}\:i\:o = \mathtt{do}~\{({\it candidate}_i,{\it cond}) \leftarrow \mathtt{Get}\:i;~({\it candidate}_o,M) \leftarrow \mathtt{Get}\:o;\\
\!\!\!&\!&~~~~~~~~~~~~~~~~~~~~~~~~~~\:\:\:\:\mathsf{if}~({\it cond}\:M)~\mathsf{then}~\mathtt{do}~\{ \\
\!\!\!&\!&~~~~~~~~~~~~~~~~~~~~~~~~~~\:\:\:\:~~~~~~~~~~\dots;~\mathtt{Put}\:{\it candidate}_i\:(\mathtt{Just}\:{\it decision}_i);~\dots;\\
\!\!\!&\!&~~~~~~~~~~~~~~~~~~~~~~~~~~\:\:\:\:~~~~~~~~~~\dots;~\mathtt{Put}\:{\it candidate}_o\:(\mathtt{Just}\:{\it decision}_o);~\dots;\\
\!\!\!&\!&~~~~~~~~~~~~~~~~~~~~~~~~~~\:\:\:\:~~~~~~~~~~\dots \}\\
\!\!\!&\!&~~~~~~~~~~~~~~~~~~~~~~~~~~\:\:\:\:\mathsf{else}~\mathtt{do}~\{\mathtt{Put}\:{\it candidate}_i\:\mathtt{Nothing};~\mathtt{Put}\:{\it candidate}_i\:\mathtt{Nothing}\} \} \\
&&\mathtt{@PointI} :: \type{Synchronizer} \rightarrow \type{Point} \rightarrow \type{In}~\tau \rightarrow (\tau \rightarrow \type{Bool}) \rightarrow \type{IO}~\tau \rightarrow \type{IO}~\tau\\ 
&&\mathtt{@PointI}\:s\:p\:i\:{\it cond}\:\alpha = \mathtt{do}~\{\dots;~\mathtt{Put}\:i\:({\it candidate},{\it cond});\\
&&~~~~~~~~~~~~~~~~~~~~~~~~~~~~~~~~~~~~~~~~~~~~\:\:\:\:\:\:x \leftarrow \mathtt{Get}\:{\it candidate}; \\
&&~~~~~~~~~~~~~~~~~~~~~~~~~~~~~~~~~~~~~~~~~~~~~\:\:\:\:\:\mathtt{maybe}~(\mathtt{@PointI}\:s\:p\:i\:{\it cond}\:\alpha) \\
&&~~~~~~~~~~~~~~~~~~~~~~~~~~~~~~~~~~~~~~~~~~~~~\:\:\:~~~~~\:\:\:\:\:~~~~(\lambda {\it decision}.~\mathtt{do}~\{\mathtt{Put}\:s\:(p,{\it decision});~\dots\})~x\}\\
&&\mathtt{@PointO} :: \type{Synchronizer} \rightarrow \type{Point} \rightarrow \type{Out}~\tau \rightarrow \tau \rightarrow \type{IO}~\unit \rightarrow \type{IO}~\unit\\ 
&&\mathtt{@PointO}\:s\:p\:o\:M\:\alpha = \mathtt{do}~\{\dots;~\mathtt{Put}\:o\:({\it candidate},M);\\
&&~~~~~~~~~~~~~~~~~~~~~~~~~~~~~~~~~~~~~~~~~~~~\:\:\:\:\:\:x \leftarrow \mathtt{Get}\:{\it candidate}; \\
&&~~~~~~~~~~~~~~~~~~~~~~~~~~~~~~~~~~~~~~~~~~~~~\:\:\:\:\:\mathtt{maybe}~(\mathtt{@PointO}\:s\:p\:o\:M\:\alpha) \\
&&~~~~~~~~~~~~~~~~~~~~~~~~~~~~~~~~~~~~~~~~~~~~~\:\:\:~~~~~\:\:\:\:\:~~~~(\lambda {\it decision}.~\mathtt{do}~\{\mathtt{Put}\:s\:(p,{\it decision});~\dots\})~x\}
\end{eqnarray*}
%
%
Finally, we make the following trivial adjustments to the type constructor $\type{channel}$, and the functions $\mathtt{receive}$ and $\mathtt{transmit}$.
$$\mathtt{type}~~\type{channel}~\tau ~=~ (\type{In}~\tau,\type{Out}~\tau,\type{MVar}~\tau) \vspace{-2.5mm}$$
\begin{eqnarray*}
&&\mathtt{receive}\:(i,o,m)\:{\it cond} =  \lambda s\:{\it name}\:{\it abort}.~\mathtt{do}~\{\dots;~\mathtt{@PointI}\:s\:p\:i\:{\it cond}\:(\mathtt{Get}\:m)\} \\
&&\mathtt{transmit}\:(i,o,m)\:M = \lambda s\:{\it name}\:{\it abort}.~\mathtt{do}~\{\dots;~\mathtt{@PointO}\:s\:p\:o\:M\:(\mathtt{Put}\:m\:M)\} 
\end{eqnarray*}

\section{Related work}\label{sec:relwork}
We are not the first to implement CML-style concurrency primitives in another language. 
In particular, Russell presents an implementation of events in Concurrent Haskell in \cite{russell}. The implementation provides \emph{guarded channels}, which filter communication based on conditions on message values  (as in Section \ref{guards}). 
Unfortunately, the implementation requires a rather complex Haskell type for $\type{event}$ values. In particular, a value of type $\type{event}~\tau$ must carry (among other things) a continuation of type $\type{IO}~\tau \rightarrow \type{IO}~\type{\unit}$. 
An important difference between Russell's implementation and ours is that Russell's $\mathtt{choose}$ combinator is asymmetric. In contrast, we implement a symmetric $\mathtt{choose}$ combinator, following the standard CML semantics. While it is difficult to compare other aspects of our implementations, we should point out that Russell's event library is more than 1300 lines of Haskell code (without comments), compared to our 150. 
Yet, guarded communication in the sense of Russell can be readily implemented in our setting, as shown in Section \ref{guards}.  
In the end, we believe that this difference in complexity is largely due to the elegance of our synchronization protocol.

Recently, Donnelly and Fluet \cite{transevents} introduce \emph{transactional events} and implement them over the software transactional memory (STM) module in Concurrent Haskell. Their key observation is that combining all-or-nothing transactions with CML-style concurrency recovers a monad. Unfortunately, implementing transactional events requires solving NP-hard problems \cite{transevents}. In contrast, our direct implementation of CML-style concurrency remains rather lightweight.

Other implementations of events include those of Flatt and Findler in Scheme~\cite{flatt-findler} and of Demaine in Java \cite{demaine}. While Flatt and Findler focus on kill-safety, Demaine focuses on efficiency by exploiting communication patterns that involve either single receivers or single senders. Demaine does not consider event combinators---in particular, it is not clear whether his implementation can accommodate abort actions.

Distributed protocols for implementing selective communication date back to the 1980s. The protocols of Buckley and Silberschatz \cite{buckley-silberschatz} and Bagrodia \cite{bagrodia} seem to be among the earliest in this line of work. Unfortunately, those protocols are prone to deadlock. Bornat \cite{bornat} proposes a protocol that is deadlock-free assuming communication between single receivers and single senders. Finally in \cite{knabe}, Knabe presents the first deadlock-free protocol to implement selective communication for arbitrary channel communication. Knabe's protocol appears to be the closest to ours. Channels are considered as sites of control, and messages are exchanged between communication points and channels to negotiate synchronization. However, Knabe assumes a global ordering on processes and maintains queues for matching points; we do not require either of these facilities in our protocol. Moreover, as in \cite{demaine}, it is not clear whether the protocol can accommodate event combinators such as $\mathtt{guard}$ and $\mathtt{wrapabort}$.



\section{Conclusion}
In this paper, we show how to implement first-class event synchronization in Concurrent Haskell, a language with first-order message passing. We appear to be the first to implement the standard semantics for events and event combinators in this setting. An interesting consequence of our work is that implementing distributed selective communication is reduced to implementing distributed message-passing in Concurrent Haskell. At the heart of our implementation is a new deadlock-free protocol that is run among communication points, channels, and synchronization sites. This protocol seems to be robust enough to allow implementations of sophisticated synchronization primitives.

All the code presented in this paper is available online at 
$$\mbox{\url{http://www.soe.ucsc.edu/~avik/projects/CML}}$$ 











\paragraph{Acknowledgments}
This work started off as a class project for Cormac Flanagan's course on \emph{Concurrent Programming and Transactional Memory} at UC Santa Cruz in Spring 2007. Thanks to him, Mart\'in Abadi, and Andy Gordon for their enthusiasm and curiosity, which encouraged further development of this work since that course. 

\bibliographystyle{abbrv}
\bibliography{ref}

\appendix


\section{Correctness proof for the synchronization protocol}

In this appendix, we prove correctness of the synchronization protocol (Theorem \ref{bisim}). This proof closely follows the informal proof of progress in the Land of Fr\={\i}g. 

Consider any state $\sigma$. We begin by defining some invariants that $\sigma$ must satisfy; the satisfaction of these invariants is written as $\vdash \sigma$. Let $\sigma$ be in the form $(\nu P)~\Pi \seq\varsigma$. We assume that the set of points in $\seq \varsigma$ is $P$, the set of channels in $\seq \varsigma$ is $\mathcal C$, the set of synchronizers in $\seq \varsigma$ is $\mathcal S$, and $\mathcal C \cap P = \varnothing$. Then $\vdash \sigma$ if:

\begin{enumerate}
\item For any $p \in \mathcal P$, there is a unique $s \in \mathcal S$ such that $p \in \dom(s)$. Further, let $s(p) = \alpha$. Then at most one of the following sub-states is in $\seq\varsigma$:
$$\{p \mapsto \alpha,\mathsf{Candidate}_p,\mathsf{Select}_s(p),\mathsf{Reject}(p),\mathsf{Done}_s(p),\mathsf{Retry}_s\}$$
and exactly one of the following sub-states is in $\seq\varsigma$:
$$\{\Box_s,\boxtimes_s\}$$
\item For every $c \in \mathcal C$, exactly one of the following sub-states is in $\seq\varsigma$:
$$\{\odot_c,\mathsf{Match}_c(\_,\_)\}$$
\item Let $p \in P$ and $s \in \mathcal S$ such that $p \in \dom(s)$. Then, if one of the following sub-states is in $\seq\varsigma$:
$$\{\mathsf{Select}_s(p),\mathsf{Reject}(p),\mathsf{Done}_s(p),\mathsf{Retry}_s\}$$
then $\boxtimes_s$ is in $\seq\varsigma$.
On the other hand, if $\boxtimes_s$ is in $\seq\varsigma$ then there is at most one $p$ such that $p \in \dom(s)$ and one of the following sub-states is in $\seq\varsigma$:
$$\{\mathsf{Select}_s(p),\mathsf{Done}_s(p),\mathsf{Retry}_s\}$$
\item Let $c \in \mathcal C$ and $p,q \in P$. Then $\mathsf{Match}_c(p,q)$ is in $\seq\varsigma$ iff there are (not necessarily distinct) $s,s' \in \mathcal S$ such that $s(p) = c$, $s'(q) = \overline c$, 
one of the following sub-states is in $\seq\varsigma$:
$$\{\mathsf{Candidate}_p,\mathsf{Select}_s(p),\mathsf{Reject}(p)\}$$ 
and one of the following sub-states is in $\seq\varsigma$:
$$\{\mathsf{Candidate}_q,\mathsf{Select}_{s'}(q),\mathsf{Reject}(q)\}$$
\end{enumerate}

It is easy to see that (1--4) are invariants for any state compiled from a program.
\begin{lemma}\label{inv} Let $\mathcal C$ be the set of channels in a program $\Pi_{i \in 1..n} S_i$. Now, suppose that $\displaystyle\Pi_{c \in \mathcal C} \odot_c~|~\Pi_{i\in 1..n} \widehat{S_i} \rightarrow^\star \sigma$. Then $\vdash\sigma$. 
\end{lemma}

Next, we prove the analog of the lemma in Section \ref{intro}. 
\begin{lemma}\label{priests} Suppose that $\vdash (\nu P)~\Pi\seq\varsigma$. Let $p,q \in P$, and let $p\mapsto c$ and $q \mapsto \overline c$ be in $\seq\varsigma$. Then $(\nu P)~\Pi\seq\varsigma \longrightarrow^\star (\nu P)\Pi\seq{\varsigma'}$ such that $\odot_c$ is in $\seq{\varsigma'}$.
\end{lemma}
\begin{proof}
If $\odot_c$ is in $\seq{\varsigma}$, we are done. Otherwise, by (2) it follows that $\mathsf{Match}_c(p',q')$ in in $\seq{\varsigma}$ for some $p',q' \in P$. Now, by (1) and (4), $p' \neq p$ and $q' \neq q$. Further, by (4), there are $s$ and $s'$ in the set of synchronizers in $\seq\varsigma$ such that $p' \in \dom(s)$, $q' \in \dom(s')$,
$\mathsf{Candidate}_{p'}$ or $\mathsf{Select}_s(p')$ or $\mathsf{Reject}(p')$ is in $\seq\varsigma$, and $\mathsf{Candidate}_{q'}$ or $\mathsf{Select}_{s'}(q')$ or $\mathsf{Reject}(q')$ is in $\seq\varsigma$. Further, by (1), $\Box_s$ or $\boxtimes_s$ is in $\seq\varsigma$, and $\Box_{s'}$ or $\boxtimes_{s'}$  is in $\seq\varsigma$. Now, {\rm (II.i--ii)} can be applied to move to a state $(\nu P)\Pi\seq{\varsigma''}$ such that $\mathsf{Select}_s(p')$ or $\mathsf{Reject}(p')$ is in $\seq{\varsigma''}$, and $\mathsf{Select}_{s'}(q')$ or $\mathsf{Reject}(q')$ is in $\seq{\varsigma''}$. Finally, {\rm (III.i--iv)} can be applied to move to the required state  $(\nu P)\Pi\seq{\varsigma'}$.
\end{proof}


We are now ready to prove the main theorem.

~

\noindent
{\bf Restatement of Theorem \ref{bisim}.}~Let $\mathcal C$ be the set of channels in a program $\Pi_{i \in 1..n} S_i$. Then $\Pi_{i \in 1..n} S_i~\sim~\displaystyle\Pi_{c \in \mathcal C} \odot_c~|~\Pi_{i\in 1..n} \widehat{S_i}$, where $\sim$ is the largest relation such that 
$\mathcal P \sim \sigma$ iff
\begin{description}
\item[{\em (Correspondence)}] $\sigma \rightarrow^\star \sigma'$ for some $\sigma'$ such that $\ulcorner \mathcal P \urcorner = \ulcorner \sigma' \urcorner$;
\item[{\em (Safety)}] if $\sigma \rightarrow \sigma'$ for some $\sigma'$, then $\mathcal P \rightarrow^\star \mathcal P'$ for some $\mathcal P'$ such that $\mathcal P' \sim \sigma'$;
\item[{\em (Progress)}] if $\mathcal P \rightarrow \_$, then $\sigma \rightarrow^+\! \sigma'$ and $\mathcal P \rightarrow \mathcal P'$ for some $\sigma'$ and $\mathcal P'$ such that  $\mathcal P' \sim \sigma'$.
\end{description}

\begin{proof} The proof of (Safety) is fairly easy. (Progress) follows from Lemmas \ref{inv} and \ref{priests}, as follows. By Lemma \ref{inv}, we can consider only a subset $\simeq$ of $\sim$ such that $\_ \simeq \sigma ~\Rightarrow~ \vdash \sigma$.  Now, suppose that $\mathcal P \simeq (\nu P)~\seq\varsigma$ and $\mathcal P \rightarrow \_$. 

We first assume that there are some $p\mapsto c$ and $q \mapsto \overline c$ in $\seq\varsigma$. Then, by Lemma \ref{priests} and {\rm (I)}, $(\nu P)~\seq\varsigma \rightarrow^\star (\nu P)~\seq{\varsigma'}$ such that $\mathsf{Candidate}_p$ and $\mathsf{Candidate}_q$ are in $\seq{\varsigma'}$. It can be shown, following the reasoning of Section \ref{intro}, that eventually: either $c$ and $\overline c$ are released; or, some $\alpha$ is released such that for some point $p'$ and synchronizer $s$, we have $s(p') = \alpha$ and either $p \in \dom(s)$ or $q \in \dom(s)$. In the former case, $\mathcal P$ makes the corresponding reduction, and we are done. In the latter case, the action complementary to $\alpha$ is released\linebreak[4] in the next step; then $\mathcal P$ makes the corresponding reduction, and we are done. 

On the other hand, if there are no $p\mapsto c$ and $q \mapsto \overline c$ in $\seq\varsigma$, then by (1) there must be some $p$ and $q$, and $s$ and $s'$, such that $s(p) = c$, $s'(q) = \overline c$, and one of the following sub-states is in $\seq\varsigma$: 
$$\{\mathsf{Candidate}_p,\mathsf{Select}_s(p),\mathsf{Reject}(p),\mathsf{Done}_s(p),\mathsf{Retry}_s\}$$
and one of the following sub-states is in $\seq\varsigma$: 
$$\{\mathsf{Candidate}_q,\mathsf{Select}_{s'}(q),\mathsf{Reject}(q),\mathsf{Done}_{s'}(q),\mathsf{Retry}_{s'}\}$$
\end{proof}
So either $\mathsf{Candidate}_p$ and $\mathsf{Candidate}_q$ are in $\seq{\varsigma}$, or $p$ and $q$ are in sub-states reachable from  $\mathsf{Candidate}_p$ and $\mathsf{Candidate}_q$; in either case, we proceed as before. 
\end{document}